\documentclass[a4paper]{article}
\usepackage{spconf,amsmath,graphicx,cite,amssymb}

\newtheorem{theorem}{Theorem}[section]

\newenvironment{proof}[1][Proof]{\begin{trivlist}
\item[\hskip \labelsep {\bfseries #1}]}{\end{trivlist}}

\newenvironment{example}[1][Example]{\begin{trivlist}
\item[\hskip \labelsep {\bfseries #1}]}{\end{trivlist}}

\newcommand{\qed}{\nobreak \ifvmode \relax \else
      \ifdim\lastskip<1.5em \hskip-\lastskip
      \hskip1.5em plus0em minus0.5em \fi \nobreak
      \vrule height0.75em width0.5em depth0.25em\fi}

\DeclareMathOperator{\rank}{rank}
\DeclareMathOperator{\diag}{diag}

\def\mB{\mbox{$\mathbf{B}$}}

\def\mD{\mbox{$\mathbf{D}$}}

\newcommand{\E}{{\cal E}}
\newcommand{\G}{{\cal G}}

\newcommand{\Real}{\mathbb R}

\newcommand{\abs}[1]{\left\vert#1\right\vert}

\def\E{\mbox{${\cal E}$}}
\def\G{\mbox{${\cal G}$}}

\def\b0{\mbox{\boldmath $0$}}

\def\bff{\mbox{\boldmath $f$}}
\def\bg{\mbox{\boldmath $g$}}

\def\bu{\mbox{\boldmath $u$}}

\def\bx{\mbox{\boldmath $x$}}

\def\bpsi{\mbox{\boldmath $\psi$}}

\def\bphi{\mbox{\boldmath $\phi$}}

\title{On the degrees of freedom of signals on graphs}
\name{Mikhail Tsitsvero and Sergio Barbarossa\thanks{This work was supported by the European Project TROPIC Project, Nr. 318784.}}
\address{Sapienza Univ. of Rome, DIET Dept., Via Eudossiana 18, 00184 Rome, Italy\\
E-mail: {\tt tsitsvero@gmail.com, sergio.barbarossa@uniroma1.it}}
\begin{document}

\maketitle
\begin{abstract}
Continuous-time signals are well known for not being perfectly localized in both time and frequency domains. Conversely, a signal defined over the vertices of a graph can be perfectly localized in both vertex and frequency domains. We derive the conditions ensuring the validity of this property and then, building on this theory, we provide the conditions for perfect reconstruction of a graph signal from its samples. Next, we provide a finite step algorithm for the reconstruction of a band-limited signal from its samples and then we show the effect of sampling a non perfectly band-limited signal and show how to select the bandwidth that minimizes the mean square  reconstruction error. 
\end{abstract}
\section{Introduction}
\label{sec:intro}
In many applications, from sensor to social networks, transportation systems and gene regulatory networks, the signals of interest are defined over the vertices of a graph. In most cases the signal domain is not a metric space, as for example with biological networks, where the vertices may be genes, proteins, enzymes, etc. This marks a fundamental difference with respect to time signals where the time domain is inherently a metric space. The last years witnessed a significant advancement in the development of signal processing tools devoted to the analysis of signals on graph \cite{shuman2013emerging}, \cite{sandryhaila2013discrete}. Of particular interest is the application of this relatively new discipline to the analysis of big data, as proposed in \cite{sandryhaila2014big}. As in conventional signal processing, a central role is played by the spectral analysis of signals on graph, which passes through the introduction of the so called Graph Fourier Transform (GFT)  \cite{shuman2013emerging}. An equivalent uncertainty principle was recently derived for signals on graphs, essentially transposing Heisenberg's methodology to signal defined over a graph \cite{agaskar2013spectral}. However, although conceptually interesting, the transposition of the uncertainty principle from continuous time (and frequency) signals to signals on graphs presents a series of shortcomings, essentially related to the fact that while time and frequency are  metric spaces, the vertex domain is not. This means  that the distance between time instants cannot be transposed into the distance between vertices, evaluated for example through number of hops. To overcome these critical aspects, in this paper we resort to an alternative definition of time and frequency spread, as proposed in the seminal works of Slepian et al. \cite{Slepian:1961:PSW}, later extended by Pearl \cite{pearl1973time} and Donoho and Stark \cite{donoho1989uncertainty}.  This alternative approach allows us to show that, differently from continuous-time signals, a signal on graph can indeed be {\it perfectly localized in both vertex and frequency domains}. Interestingly, the framework necessary to prove  this perfect localization property paves the way to derive a sampling theorem for signals on graph. Building on this theory, we derive the necessary and sufficient condition and closed form expression for the perfect reconstruction of a band-limited graph signal from its samples. Finally, numerical results based on the real-world data, non necessarily band-limited, show how the developed framework may be applied for evaluating the  bandwidth that minimizes the reconstruction error.

\noindent{\bf Notation: } We consider an undirected graph $\G = (V, \E)$ consisting of a set of $N$ nodes $V = \{1,2,..., N\}$ along with a set of weighted edges $\E=\{a_{ij}\}, \ i, j \in V$, such that $a_{ij}>0$ if there is a link between nodes $i$ and $j$. The adjacency matrix $\mathbf{A}$ of the graph is the collection of all the weights $a_{ij}, i, j = 1, \ldots, N$. The degree of node $i$ is $d_i:=\sum_{j=1}^{N}a_{ij}$. The degree matrix is a diagonal matrix having the node degrees on its diagonal: $\mathbf{K} = \diag \{ d_1, d_2, ... , d_N \}$.  The combinatorial Laplacian matrix is defined as $\mathbf{L} = \mathbf{K}-\mathbf{A}$. In the literature it is also common to use normalized graph Laplacian matrix $\mathbf{\mathcal{L}}=\mathbf{K}^{-1/2}\mathbf{L}\mathbf{K}^{-1/2}$. We will use the combinatorial Laplacian matrix in the further derivations, but the approach is not limited to this case and alternative definitions could be used depending on the applications.

Since the Laplacian matrix of an undirected graph is symmetric and positive semi-definite matrix, it may be diagonalized as
\begin{equation}
\label{slep_func:laplacian_eigendecomposition}
\mathbf{L} = \mathbf{U} \mathbf{\Lambda} \mathbf{U}^T = \sum_{i=1}^N \lambda_i \bu_i \bu^T_i,
\end{equation}
where $\mathbf{\Lambda}$ is a diagonal matrix with real non-negative eigenvalues on its diagonal and $\{\bu_i \}$ is the set of real-valued orthonormal eigenvectors. A signal $\bx$ over a graph $\G$ is defined as a mapping $\bx : V \rightarrow \mathbb{R}^{\abs{V}} $. Laplacian eigenvectors endow a graph with a natural criterion of smoothness, as in the continuous time case, where the eigenfunctions of the Laplacian operator constitute the  basis of conventional Fourier transform. 

The Graph Fourier Transform is defined as a projection operator onto the space spanned by the Laplacian eigenvectors \cite{shuman2013emerging}:
\begin{equation}
\label{slep_func:gft}
\hat{\bff} = \mathbf{U}^T \bff
\end{equation}
with inverse
\begin{equation}
\label{slep_func:gft_inverse}
\mathbb{\bff} = \mathbf{U} \hat{\bff}.
\end{equation}
Given a subset of vertices $S \subseteq V$, we define a vertex-limiting operator as a diagonal matrix  $\mathbf{D}$ whose generic diagonal entry $D_{ii}$ is equal to one if $i \in S$ or zero, otherwise. Similarly, given a subset of frequency indices $F \subseteq V^*$, where $V^*=\{1, \ldots, N\}$, we introduce  a filtering operator
\begin{equation}
\label{lowpass_operator}
\mathbf{B} = \mathbf{U \Sigma U}^T,
\end{equation}
where $\mathbf{\Sigma}$ is a diagonal matrix whose diagonal entry $\Sigma_{ii}$ is equal to one if $i \in F$ or zero, otherwise. It is immediate to check that both matrices $\mathbf{D}$ and $\mathbf{B}$ are symmetric and idempotent, and then orthogonal projectors. We introduce also the complement vertex-limiting operator $\mathbf{D}^c$, defined as $\mathbf{D}$ but operating over the complement set $S^c=V\setminus S$, and the frequency-limiting complement operator $\mathbf{B}^c$ defined as $\mathbf{B}$, but operating over $F^c=V\setminus F$. 

\section{Localization Properties}
\label{Degrees of freedom}

%

The study of time and band-limited signals led Slepian, Pollak and Landau to the well-known prolate spheroidal wave functions, which in turn clarified a number of fundamental questions concerning degrees of freedom, uncertainty principle and signal approximation. In this work we follow a similar path, inspired by the seminal work of Slepian \cite{slepian1976bandwidth} and subsequent  generalization by Donoho and Stark \cite{donoho1989uncertainty} to the case when time and bandwidth are not intervals.

Following Slepian \cite{slepian1976bandwidth}, we start building an optimal set of GFT bandlimited vectors $\bpsi_i$, $i=1, \ldots, n$ that are maximally concentrated over some subset $S$ in vertex domain. The vectors $\bpsi_i$ can be found as the solution of the following iterative optimization problem
\begin{equation}
\label{slep_func:max_problem}
\begin{aligned}
\bpsi_i &&= \  & \underset{\bpsi_i}{\arg \max \,}
 \| \mathbf{D} \bpsi_i \|_2 \\
&&& \text{subject to} \\ 
&&& \| \bpsi_i\|_2 = 1,\\ 
&&& \mathbf{B} \bpsi_i = \bpsi_i,\\
&&& \text{if} \ i>1,\, \bpsi^T_i \bpsi_j = 0, \ j = 1,...,i-1.
\end{aligned}
\end{equation}
In this way, $\bpsi_1$ is the band-limited vector maximally concentrated on $S$, $\bpsi_2$ is the band-limited vector belonging to the subspace orthogonal to $\bpsi_1$, which is maximally concentrated on $S$, and so on. In this way, we build a set of orthonormal bandlimited vectors whose energies are maximally concentrated over $S$. These vectors are the counterpart of the prolate spheroidal wave functions introduced by Slepian and Pollack in \cite{Slepian:1961:PSW}.

Because of the band-limiting constraint, problem (\ref{slep_func:max_problem}) does not change if we substitute $\mathbf{D}$ with $\mathbf{DB}$ in the objective function. The problem can then be reformulated as follows
\begin{equation}
\label{slep_func:max_problem_restated}
\begin{aligned}
\bpsi_i = \ & \underset{\bpsi_i}{\arg \max}
& & \| \mathbf{DB} \bpsi_i \|_2 \\
& \text{subject to}
& & \|\bpsi_i\|_2 = 1,\\
&&& \text{if} \ i>1,\, \bpsi^T_i \bpsi_j = 0, \ j = 1,...,i-1.
\end{aligned}
\end{equation}
Using Rayleigh-Ritz theorem, the solutions of (\ref{slep_func:max_problem_restated})  are known to be the eigenvectors of $\left(\mathbf{DB}\right)^T \mathbf{DB} = \mathbf{BDB}$, i.e.
\begin{equation}
\label{slep_func:bdb_eigendecomposition}
\mathbf{BDB} \bpsi_i = \lambda_i \bpsi_i.
\end{equation}
Clearly, the vectors $\bpsi_i$ are bandlimited by construction and hence they correspond to solutions of (\ref{slep_func:max_problem}). The set of eigenvectors $\bpsi_i$ (we will refer to them as Slepian vectors on graphs) constitute a basis for all GFT bandlimited functions. Moreover, they are orthogonal on the sampling subset
\begin{equation}
\label{slep_func:orthogonality_on_sampling_set}
\bpsi_i^T \mathbf{D} \bpsi_j = \lambda_j \delta_{ij},
\end{equation}
where $\delta_{ij}$ is the Kronecker symbol.
Equation (\ref{slep_func:orthogonality_on_sampling_set}) follows from bandlimitedness of $\bpsi_i$ and from the self-adjointness of $\mathbf{B}$. 


Let us consider now the question if a signal on graph can be perfectly localized over a vertex set $S$ and a frequency set $F$. The answer is given by the following
\begin{theorem} 
\label{theorem_unit_eigenvalue}
There is a vector $\bx$, perfectly localized over both vertex set $S$ and frequency set $F$, if and only if the matrix $\mathbf{B} \mathbf{D}\mathbf{B}$ has an eigenvalue equal to one; in such a case, $\bx$ is an eigenvector associated to the unit eigenvalue.
\end{theorem}
\begin{proof}:
If a vector $\bx$ is perfectly localized in both vertex and frequency domains, then the following equalities hold: 
\begin{equation}
\label{Bx=x;Dx=x}
\mB \bx=\bx, \,\,\,\mD \bx=\bx.
\end{equation}
Now we prove that, if these conditions hold, then $\bx$ must be an eigenvector of $\mB \mD \mB$ associated to a unit eigenvalue. Indeed, in such a case, by repeated applications of (\ref{Bx=x;Dx=x}), it follows
\begin{equation}
\label{BDBx = BDx}
\mB \mD \mB \bx = \mB \mD \bx = \mB \bx = \bx.
\end{equation}
This proves the first part. Now, let us prove that, if $\bx$ is an eigenvector of $\mB \mD \mB$ associated to a unit eigenvalue, then $\bx$ must respect (\ref{Bx=x;Dx=x}). Indeed, starting from
\begin{equation}
\label{BDBx=x}
\mB \mD \mB \bx = \bx
\end{equation}
and multiplying from the left side by $\mB$, taking into account that $\mB^2=\mB$, we get
\begin{equation}
\label{BDBx=Bx}
\mB \mD \mB \bx = \mB \bx
\end{equation}
Equating (\ref{BDBx=x}) to (\ref{BDBx=Bx}), we get
\begin{equation}
\mB \bx= \bx,
\end{equation}
which implies that $\bx$ is perfectly localized in the frequency domain. Now, using this property, we can also write
\begin{equation}
1=\max_{\bx} \frac{\bx^T \mB \mD \mB \bx}{\bx^T \bx}=\max_{\bx} \frac{\bx^T \mD \bx}{\bx^T \bx}
\end{equation}
This shows that $\bx$ is also an eigenvector of $\mathbf{D}$ associated to a unit eigenvalue, i.e., $\bx$ is also perfectly localized in the vertex domain.
\qed
\end{proof}
Equivalently, since $\sigma_i(\mB \mD) = \sigma_i(\mD \mB)$, perfect localization onto the sets $S$ and $F$ is achieved if the following properties hold true:
\begin{equation}
\label{|BD|=1=|DB|}
\|\mB \mD\|_2 = 1; \,\,\,\, \|\mD \mB\|_2 = 1.
\end{equation}
If a vector $\bx$ is perfectly localized over $S$ and $F$, then 
\begin{equation}
\label{BcDx=0}
\mB^c \mD \bx=\b0.
\end{equation}
Hence, since $\mathbf{U}$ is a unitary matrix, perfect localization is feasible only if the following system of linear equations admits a non-trivial solution
\begin{equation}
\label{Sigma_U_D}
\mathbf{\Sigma}^c \mathbf{U}^T \mathbf{D} \bphi = \b0.
\end{equation}
This equation can be rewritten in compact form by retaining only the nonzero rows corresponding to the indices $i$ belonging to the complement set $F^c$, and the columns corresponding to the indices in $S$. More specifically, denoting by $i_1, \ldots, i_{|F^c|}$ the indices in $F^c$ and by $j_1, \ldots, j_{|S|}$ the indices in $S$, (\ref{Sigma_U_D}) can be rewritten as
\begin{equation}
\label{slae_homogeneous}
\mathbf{G} \bphi = \b0,
\end{equation}
where $G_{k\ell} = u_{i_k} (j_{\ell})$, for $k=1, \ldots, N-|F|$ and  $\ell=1, \ldots, |S|$, is the $j_{\ell}$ entry of the $i_k$-th column of  $\mathbf{U}$.
Matrix $\mathbf{G}$ has the dimensionality $(N - |F|) \times |S|$. Clearly, if
\begin{equation}
\label{Perfect_localization}
|S| \geq N-|F| +1
\end{equation}
system (\ref{slae_homogeneous}) admits a nontrivial solution. The inequality (\ref{Perfect_localization}) is then the condition for {\it perfect localization} in both vertex and frequency domains. This is indeed a remarkable property that marks a fundamental difference with respect to continuous-time signals. Conversely, if $|S| \leq (N - |F|)$,  system (\ref{slae_homogeneous}) can still admit a non-trivial solution, but only if the matrix $\mathbf{G}$ is not full column rank.

For any given pair of sets $S$ and $F$, the dimension of the signal subspace of signals perfectly localized over both vertex and GFT sets is given by the following theorem.
\begin{theorem}[Number of the degrees of freedom]
\label{theorem::dof}
The number of singular values equal to 1 of the operator $\mathbf{B} \mathbf{D}$ is equal to 
\begin{equation}
\label{number_of_ones}
C := \rank \mathbf{D}-\rank \mathbf{B}^c \mathbf{D}.
\end{equation}
\end{theorem}
\begin{proof}
The proof  follows the arguments just reported above. \qed
\end{proof}

%
In general, the number of singular values from the "transition region", i.e. $0 < \lambda_i < 1$  is 
\begin{align}
\label{dof:num_of_decaying}
Q := 
\begin{cases}
\rank \mathbf{B} \mathbf{D}^c, \ \text{if} \ \rank \mathbf{D} \geq \rank \mathbf{B}; \\
\rank \mathbf{B}^c \mathbf{D}, \ \text{if} \ \rank \mathbf{D} \leq \rank \mathbf{B},
\end{cases}
\end{align}
and the number of singular values equal to zero must be equal to $O := N - C - Q$.

\section{Sampling signals on graphs}
\label{ssec:sampling}

Graph signals may represent various phenomena and in many practical applications the acquisition of signals on a high number of vertices may be too expensive or too lenghty. This kind of problems pushes us to consider the reconstruction of bandlimited functions on graphs from a limited number of samples. Considerable amount of work has been done recently regarding the problem of GFT bandlimited signal reconstruction from incomplete sampled measurements. 
There are three major approaches. One is based on frame theory, mainly pursued by \cite{pesenson2008sampling}, \cite{wang2014local}, addressing the problem of finding conditions for the existence of dual frames for reconstruction from sampled data. A second approach, see e.g. \cite{wang2014local}, \cite{narang2013signal}, looks for iterative reconstruction algorithms. A third approach applies the recently developed machinery of compressive sensing to the problem of graph sampling and it was followed by \cite{zhu2012approximating}. In this work,  we provide an algorithm which allows perfect reconstruction of a bandlimited signal with a finite number of steps.

To state the problem formally, we want to reconstruct a signal $\bff$ defined on graph $\G$ by samples taken from a subset $S$, i.e., from 
\begin{equation}
\bff_S := \mathbf{D} \bff.
\end{equation}
The signal $\bff$ is supposed to be bandlimited in the sense that its support in the GFT domain is limited to the set $F$. Next we provide a theorem which gives us the necessary and sufficient condition for reconstruction.

\begin{theorem} [Sampling Theorem]
\label{theorem::sampling_theorem}
Let $F$ be the set of frequencies and $S$ be the sampling set of vertices. It is possible to reconstruct a signal, bandlimited to the set $F$ in GFT domain, only from its values on the sampling set $S$, if and only if 
\begin{equation}
\label{sampling_theorem_condition}
\|\mathbf{B} \mathbf{D}^c  \|_2 = \| \mathbf{D}^c \mathbf{B} \|_2 < 1,
\end{equation}
i.e. the operator $ \mathbf{B} \mathbf{D}^c$ does not have any perfectly localized eigenvector.
\end{theorem}

\begin{proof}
The sampled signal $\bff_S$ can be rewritten as $\bff_S=\mD \bff=(\mathbf{I}-\mD^c)\bff=(\mathbf{I}-\mD^c\mB)\bff$, where in the last equality we exploited the band-limiting property of $\bff$. The recovery of $\bff$ from $\bff_S$ is possible if the matrix $\mathbf{I}-\mD^c\mB$ is invertible. This is of course possible if (\ref{sampling_theorem_condition}) holds true. 
On the other side, if $\|\mathbf{B} \mathbf{D}^c  \|_2 = 1$, then from Theorem \ref{theorem_unit_eigenvalue} it follows that there exists a bandlimited signal perfectly localized on $S^c$. Therefore if we sample such signal on $S$ - we will get only zero values and it will be impossible to reconstruct non-zero values on $S^c$. \qed
\end{proof}
Alternatively, taking into account Theorem \ref{theorem::dof}, $S^c$ does not contain any perfectly localized bandlimited signal if
\begin{equation}
\label{condition_sampling}
\rank \mathbf{D}^c = \rank \mathbf{B}^c \mathbf{D}^c.
\end{equation}
Next we provide theorem which provides a formula for reconstruction of a bandlimited signal from its samples.
\begin{theorem}
\label{theorem::necessary_and_sufficient}
If the condition of the sampling theorem holds true, i.e.
\begin{equation}
\| \mathbf{B} \mathbf{D}^c \|_2 < 1,
\end{equation}
then any signal bandlimited to the set of frequencies $F$ can be reconstructed from its samples on the set $S$ using the following reconstruction formula
\begin{equation}
\label{eq::sampling_theorem_formula}
\bff = \sum_{i=1}^{|F|} \frac{1}{\sigma_i^2} \langle \mD \bff, \bpsi_i \rangle \bpsi_i,
\end{equation}
where $\left\{ \bpsi_i \right\}_{i = 1..N}$ and $\left\{ \sigma^2_i \right\}_{i = 1..N}$ are the eigenvectors and eigenvalues of $\mathbf{B} \mathbf{D} \mathbf{B}$.
\end{theorem}
\begin{proof}
Since $\left\{ \bpsi_i \right\}_{i = 1..N}$ constitute a basis for $\Real^N$, we can write, for any $\bg$,
\begin{equation}
\bg = \sum_{i = 1}^{N} \langle \bg, \bpsi_i \rangle \bpsi_i,
\end{equation}
and for its bandlimited projection
\begin{equation}
\mB \bg = \sum_{i = 1}^K \langle \mB \bg, \bpsi_i \rangle \bpsi_i,
\end{equation}
where $|F| \leq K \leq N$. By the condition of theorem there is no any perfectly localized bandlimited vector, therefore all the eigenvectors from $\ker \mB \mD \mB$ are out of band $F$ and we conclude that $K = |F|$.
Let $\bff = \mB \bg$, then we can write
\begin{equation}
\bff = \sum_{i = 1}^{|F|} \langle \bff, \frac{1}{\sigma^2_i} \mB \mD \mB \bpsi_i \rangle \bpsi_i =
\sum_{i = 1}^{|F|} \frac{1}{\sigma^2_i}  \langle \mD \bff, \bpsi_i \rangle \bpsi_i,
\end{equation}
where we have used the fact that operators $\mD$ and $\mB$ are self-adjoint and vectors from the set $\left\{ \bpsi_i \right \}_{i=1..|F|}$ are bandlimited. \qed
\end{proof}

The interesting new aspect is the role played by the graph topology. In particular, besides choosing the right number of samples, for any given GFT bandwidth, in order to fulfill the previous perfect reconstruction conditions, a particular relevant question is how to choose the sample vertices in the graph.  Some hints can come from (\ref{sampling_theorem_condition}). One could choose the sample vertices so that the maximum singular value of $\mathbf{B}\mathbf{D}^c$ is minimum. Or, one could look at the set of sample vertices that maximize the conditioning of matrix $\mathbf{G}$ in (\ref{slae_homogeneous}). These are, in general, combinatorial problems, and then NP-hard. However, the structure of the GFT basis vectors can provide useful hints. In particular, for bandlimited signals lying in the signal subspace spanned by a number of GFT basis vectors, looking at the amplitudes of the basis vectors entries is a good indication on which vertices to select. This is indeed an aspect worth of further investigation. 


Real world signals usually are not exactly band-limited. In general, any signal can be decomposed in two parts
\begin{equation}
\bff = \mB \bff + \mB^c \bff.
\end{equation}
We wish to assess now the impact of signal's non bandlimitedness on the reconstruction of signal from its samples. Let $\bff$ be the initial signal and $\tilde{\bff}$ be its reconstruction given by (\ref{eq::sampling_theorem_formula}), i.e.
\begin{equation}
\tilde{\bff} = \mB \bff + \sum_{i \in F} \frac{1}{\sigma_i^2} \langle \mD \mB^c \bff, \bpsi_i \rangle \bpsi_i.
\end{equation}
The reconstruction error in this case is given by
\begin{align}
\label{eq::non_bandlim_mse}
\| \bff - \tilde{\bff} \|_2 &= \left\| \sum_{i \in F^c}  \langle \bff, \bpsi_i \rangle \bpsi_i - \sum_{i \in F} \frac{1}{\sigma_i^2} \langle \mD \mB^c \bff, \bpsi_i \rangle \bpsi_i \right\|_2 \nonumber \\
&= \sum_{i \in F^c} \abs{\langle \bff, \bpsi_i \rangle}^2 + \sum_{i \in F} \frac{1}{\sigma_i^4} \abs{\langle \mD \mB^c \bff, \bpsi_i \rangle}^2.
\end{align}
In (\ref{eq::non_bandlim_mse}) the first term corresponds to the error given by the out of band energy of signal and the second term gives the aliasing error, i.e. an in-band error introduced by the out of band components.

Some numerical results, reported below, show that there is an optimal value in choosing the signal bandwidth with respect to  the reconstruction from its samples.
\begin{figure}[t]
\centering
\includegraphics[width=7cm]{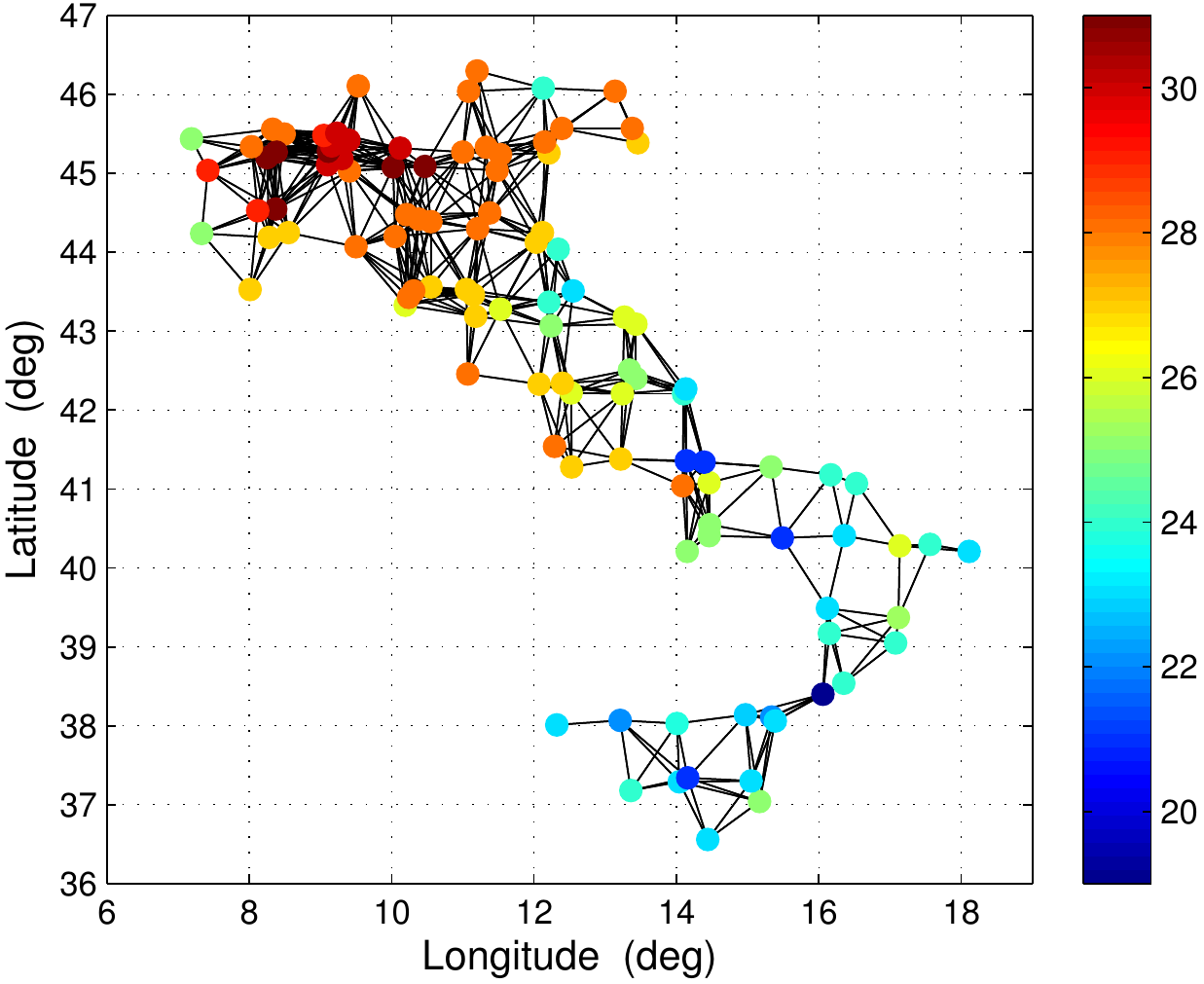}
\caption{Graph signal associated to 106 meteo stations in Italy.}\label{Italy}
\end{figure}

\begin{example}
As a numerical example, we consider the real temperature data given by 106 meteorological stations in Italy. In Fig. \ref{Italy} the corresponding graph signal is illustrated. Each graph vertex corresponds to a station and is colored with respect to the temperature registered on June 6, 2015. The graph was formed according to the geographic distance between points: there is a link between two vertices if the distance between two stations is smaller than a given coverage radius.

Based on this real-data example, it is useful to evaluate the impact of choosing the signal bandwidth on the reconstruction error. Indeed the illustrated graph signal is not exactly bandlimited, so that sampling theorem does not hold in this case. In Fig. \ref{MSE_rand} we show  the  normalized mean squared error (NMSE) as a function of the bandwidth, in accordance with (\ref{eq::non_bandlim_mse}), or, in other words, on the cardinality of $F$ for different sizes of the sampling set $S$. For each curve, parameterized to the cardinality of $S$, the  samples were chosen randomly and the results have been averaged over 500 independent selections of sampling sets. The simulations show that for this signal there exist optimal choices of the cardinality of $F$ that minimize the reconstruction error.
\end{example} 
%

\section{Conclusions}
In this work, we derived the conditions for perfect localization of a signal in both vertex and GFT domains. These conditions are strictly related to the conditions ensuring perfect reconstruction of signal from sparse samples. Interesting new developments concern the connection of localization properties in the joint vertex-frequency domain with the reconstruction from sampled observations. The effect of out of band energy on sampling was considered and the tool for assessing an optimal bandwidth was given. The results were supported by an example running on real data.
\begin{figure}[t]
\centering
\includegraphics[width=7cm]{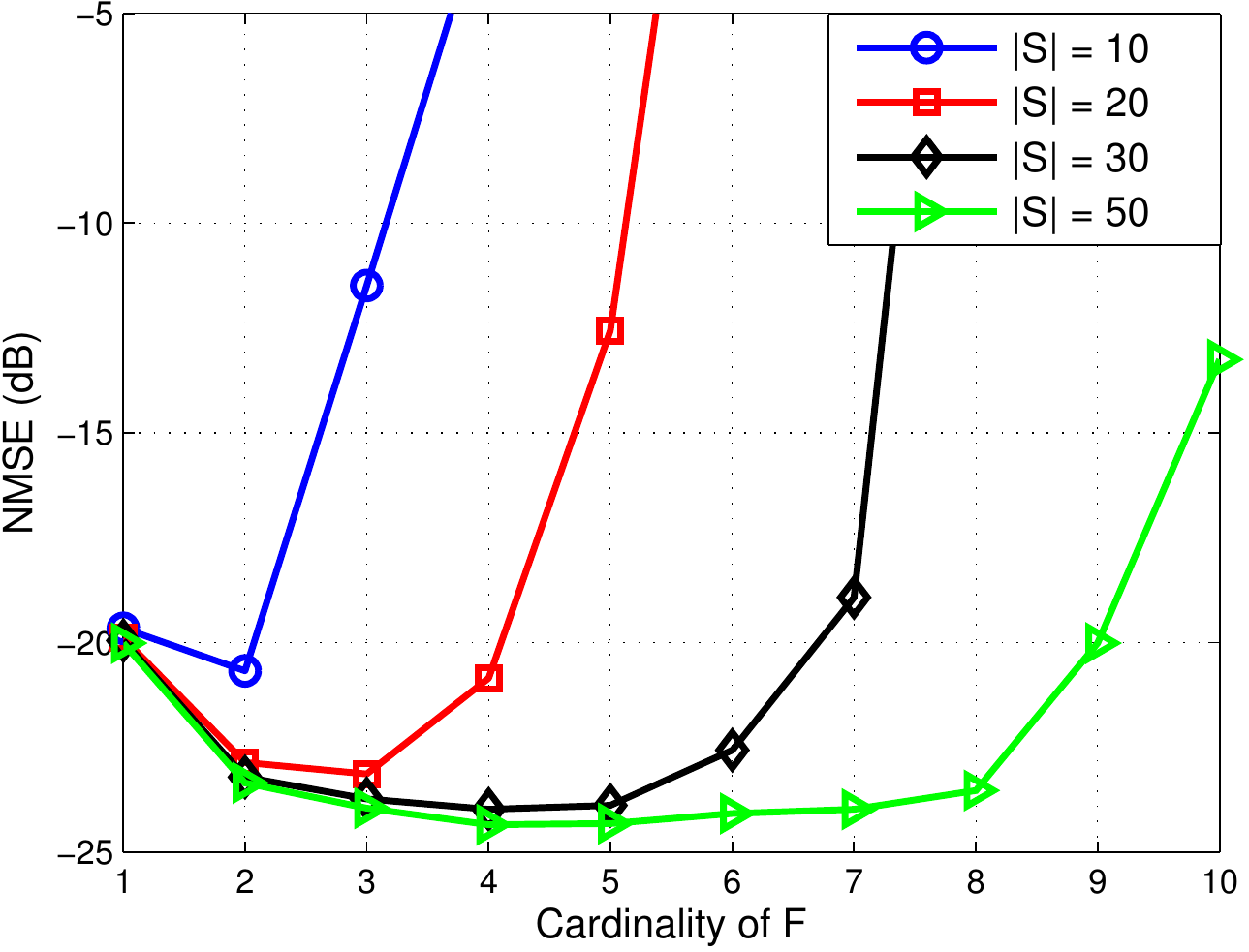}
\caption{Normalized mean squared error versus bandwidth used for processing, for different cardinalities of sampling set.}\label{MSE_rand}
\end{figure}
 
\bibliographystyle{MyIEEE}
\bibliography{refs}

\end{document}